\documentclass{lmcs}
\pdfoutput=1

\usepackage[utf8]{inputenc}
\usepackage{lastpage}
\lmcsdoi{17}{4}{4}
\lmcsheading{}{\pageref{LastPage}}{}{}%
{Aug.~05,~2020}{Oct.~29,~2021}{}


\usepackage{pdfpages}
\usepackage{wrapfig}
\usepackage[ruled,linesnumbered,noend]{algorithm2e}
\usepackage{multicol}
\usepackage{amssymb}
\usepackage{amsmath,bbm}
\usepackage{latexsym,epic,eepic}
\usepackage{multirow}
\usepackage{algorithmic}
\usepackage{url}
\usepackage{tikz}
\usetikzlibrary{shapes}
\usetikzlibrary{automata,positioning}

\newcommand{\OMIT}[1]{}

\newcommand{\defn}[1]{\textit{#1}} 
\newcommand{\mcl}[1]{\ensuremath{\mathcal{#1}}}

\newcommand{\struct}{\mathfrak{S}}

\newcommand{\Trans}[1]{\struct_{#1}}
\newcommand{\transysDom}{S}
\newcommand{\transys}{\ensuremath{\langle \transysDom; \to\rangle}}

\newcommand{\PAD}{\ensuremath{\text{PAD}}}
\newcommand{\divides}{|}
\newcommand{\PA}{\ensuremath{\mcl{P}}}
\newcommand{\ialphabet}{\ensuremath{A}}

\newcommand{\N}{\ensuremath{\mathbb{N}}}
\newcommand{\B}{\ensuremath{\mathbb{B}}}

\newcommand{\Z}{\ensuremath{\mathbb{Z}}}

\newcommand{\vecX}{\ensuremath{\text{\rm \bfseries x}}}
\newcommand{\vecY}{\ensuremath{\text{\rm \bfseries y}}}
\newcommand{\vecZ}{\ensuremath{\text{\rm \bfseries z}}}

\newcommand{\vecW}{\ensuremath{\text{\rm \bfseries w}}}
\newcommand{\vecV}{\ensuremath{\text{\rm \bfseries v}}}

\theoremstyle{definition}
 
\newtheorem{remark}[thm]{Remark} 

\newcommand{\problemx}[3]{
\par\noindent\underline{\sc#1}\par\nobreak\vskip.2\baselineskip
\begingroup\clubpenalty10000\widowpenalty10000
\setbox0\hbox{\bf Instance: }\setbox1\hbox{\bf Question: }
\dimen0=\wd0\ifnum\wd1>\dimen0\dimen0=\wd1\fi
\vskip-\parskip\noindent
\hbox to\dimen0{\box0\hfil}\hangindent\dimen0\hangafter1\ignorespaces#2\par
\vskip-\parskip\noindent
\hbox to\dimen0{\box1\hfil}\hangindent\dimen0\hangafter1\ignorespaces#3\par
\endgroup}

\newcommand{\CA}{\ensuremath{\mathcal{C}}} 
\newcommand{\Aut}{\ensuremath{\mathcal{A}}} 
\newcommand{\AutB}{\ensuremath{\mathcal{B}}} 
\newcommand{\transrel}{\ensuremath{\Delta}} 

\newcommand{\controls}{\ensuremath{Q}} 
\newcommand{\finals}{\ensuremath{F}} 
\newcommand{\Lang}{\ensuremath{\mathcal{L}}} 
\newcommand{\AutRun}{\ensuremath{\rho}} 

\newcommand{\To}{\Rightarrow}

\newcommand{\counters}{\ensuremath{X}}
\newcommand{\ctr}{\ensuremath{x}}

\theoremstyle{thm}
\newtheorem{theorem}[thm]{Theorem}
\newtheorem{proposition}[thm]{Proposition}
\newtheorem{lemma}[thm]{Lemma}
\theoremstyle{definition}
\newtheorem{example}[thm]{Example}

\newcommand{\ID}{\ensuremath{\text{\texttt{ID}}}}
\newcommand{\SUB}{\ensuremath{\text{\texttt{SUB}}}}
\newcommand{\DEC}{\ensuremath{\text{\texttt{DEC}}}}
\newcommand{\LEN}{\ensuremath{\text{\textsc{Len}}}}

\newcommand{\ACCELERATE}{\ensuremath{\text{\textsc{L-Accelerate}}}}

\newcommand{\pre}{\ensuremath{\mathit{pre}}}

\def\soln{{\sigma}}
\def\Vars{{V}}
\def\set#1{{\{ #1 \}}}

\SetCommentSty{mycommfont}
\newlength{\commentWidth}
\newcommand\lmcschanged[1]{{{#1}}}
\newif\ifdraft\draftfalse
\ifdraft
\newcommand\anthony[1]{{\color{blue}
[#1 - \textbf{Anthony}]}}
\newcommand\rupak[1]{{\color{magenta}
[#1 - \textbf{Rupak}]}}

\newcommand\todo[1]{}
\else
\newcommand\anthony[1]{}
\newcommand\rupak[1]{}

\newcommand\todo[1]{}
\fi

\newcommand\shortlong[2]{#2}

\usepackage{color}

\title[Quadratic word equations with length constraints]{
Quadratic Word Equations with Length Constraints, Counter Systems, and 
Presburger Arithmetic with Divisibility 
}
\titlecomment{{\lsuper*}An extended abstract of the paper has appeared in
ATVA 2018.}

\author[A.~W. Lin]{Anthony W. Lin}[a]
\address{TU Kaiserslautern and Max-Planck Institute for Software Systems, 
Kaiserslautern, Germany}
\email{lin@cs.uni-kl.de}

\author[R.~Majumdar]{Rupak Majumdar}[b]
\address{Max Planck Institute for Software Systems, Kaiserslautern, Germany}
\email{rupak@mpi-sws.org}

\begin{document}

\begin{abstract}
    Word equations are a crucial element in the theoretical foundation of 
constraint solving over strings.
A word equation relates two words over string variables and constants. 
Its solution amounts to a function mapping variables to constant strings that 
equate the left and right hand sides of the equation. 
While the problem of solving word equations is decidable, 
the decidability of the problem of solving a word equation
with a length constraint (i.e., a constraint relating the lengths of 
words in the word equation) has remained a long-standing open problem.
We focus on the subclass of quadratic word equations, i.e., 
in which each variable occurs at most twice. 
We first show that the length abstractions of solutions to quadratic
word equations are in general not Presburger-definable.
We then describe a class of counter systems with Presburger transition relations
which capture the length abstraction of a quadratic word equation with regular
constraints.
We provide an encoding of the effect of a simple loop of the counter systems
in the existential theory of Presburger Arithmetic with divisibility (PAD).

Since PAD is decidable (NP-hard and is in NEXP), we obtain a decision procedure 
for quadratic words equations with
length constraints for which the associated counter system is \emph{flat} (i.e.,
all nodes belong to at most one cycle).
In particular, we show a decidability result (in fact, also an NP algorithm with a PAD oracle) 
for a recently proposed NP-complete fragment of word equations called regular-oriented word equations,
when augmented with length constraints.  
We extend this decidability result (in fact, with a complexity upper bound of 
PSPACE with a PAD oracle) in the presence of regular constraints. 

\end{abstract}
\maketitle

\section{Introduction}
\label{sec:intro}

Reasoning about strings is a fundamental problem in computer science and mathematics.
The first order theory over strings and concatenation is undecidable.
A seminal result by Makanin \cite{Makanin} (see also \cite{Diekert,Jez})
shows that the satisfiability problem for the \emph{existential fragment} is decidable,
by giving an algorithm for the satisfiability of \emph{word equations}. 
A word equation $L = R$ consists of two words $L$ and $R$ over an alphabet of 
constants and variables.
It is satisfiable if there is a mapping $\sigma$ from the variables to strings over the
constants such that $\sigma(L)$ and $\sigma(R)$ are syntactically identical.
\OMIT{
Over the years, Makanin's decidability result has been improved in different ways,
to find alternate, simpler, algorithms, to characterize its complexity, 
and to generalize the result to broader classes, for example, in the presence of
regular constraints on the solutions \cite{Schulz}.
It is known today that the problem is in PSPACE and NP-hard \cite{Diekert,Plandowski,Jez}.
}

An original motivation for studying word equations was to show undecidability of
Hilbert's 10th problem (see, e.g., \cite{Matiyasevich68}). 
\lmcschanged{This was motivated
among others by an earlier result by Quine \cite{quine} 
that first-order theory over word equations essentially coincides with
first-order theory of arithmetic.}
While Makanin's later result shows that word equations could not, by themselves, show undecidability, Matiyasevich in 1968
considered an extension of word equations with \emph{length constraints} as a possible route to showing undecidability of Hilbert's
10th problem \cite{Matiyasevich68}.
A length constraint constrains the solution of a word equation by requiring a
linear relation to
hold on the lengths of words in a solution $\sigma$ (e.g., $|x| = |y|$, where
$|\cdot|$ denotes the string-length function).
The decidability of word equations with length constraints remains open. 

In recent years, reasoning about strings with length constraints has found renewed interest through applications
in program verification and reasoning about security vulnerabilities.
The focus of most research has been on developing practical string solvers
(cf.~\cite{fmsd14,S3,BTV09,Abdulla14,cvc4,HAMPI,Z3-str3,popl18-efficient,fang-yu-circuits,Berkeley-JavaScript}).
These solvers are sound but make no claims of completeness.
Relatively few results are known about the decidability status of strings with length and other constraints
(see \cite{Ganesh-boundary} for an overview of the results in this area). 
The main idea in most existing decidability results is the encoding of length
constraints into Presburger arithmetic
\cite{Ganesh13,Ganesh-boundary,Abdulla14,LB16}.
However, as we shall see in this paper, the length abstraction of a word
equation (i.e.~the set of possible lengths of variables in its solutions)
need not be Presburger definable.
\OMIT{
(Indeed, this was Matiyasevich's motivation in studying this problem as a way to prove undecidability of Hilbert's 10th problem.)
}

In this paper, we consider the case of \emph{quadratic} word equations, in which each variable can appear at most twice \cite{Lentin,diekert-quadratic1},
together with length constraints and \emph{regular constraints} (conjunctions
$\bigwedge_{i=1}^n x \in L_i$ of assertions that the variable $x$ must be
assigned a string in the regular language $L_i$ for each $i$).
For quadratic word equations, there is a simpler decision procedure (called the Nielsen transform or Levi's method) based
on a non-deterministic proof graph construction.
The technique can be extended to handle regular constraints \cite{diekert-quadratic1}.
However, we show that already for this class (even for a simple equation like
$xaby = yabx$, where $x,y$ are variables and $a,b$ are constants), the length 
abstraction need not be Presburger-definable.
Thus, techniques based on Presburger encodings are not sufficient to prove decidability.

Our first observation in this paper is a connection between the problem of quadratic word equations with length constraints
and a class of \lmcschanged{terminating} counter systems with Presburger 
transitions.
Informally, the counter system has control states corresponding to the nodes of
the proof graph constructed by Levi's method,
and a counter standing for the length of a word variable. 
\lmcschanged{Each counter value either remains the same or gets decreased
in each step of Levi's method. When all of the counters remain the same,
the counter system goes to a state, from which the previous state can no longer
be visited.}
Thus, from any initial state, the counter system terminates.
We show that the set of initial counter values which can lead to a successful leaf (i.e., one containing the trivial equation $\epsilon = \epsilon$)
is precisely the length abstraction of the word equation.

Our second observation is that 
the reachability relation for a simple loop of the counter system can be encoded in 
the existential theory of Presburger arithmetic with divisibility ($\PAD$).
%
As $\PAD$ is decidable \cite{Lipschitz,LOW15}, we obtain a technique
to symbolically represent the reachability relation for \emph{flat} counter systems,
in which each node belongs to at most one loop.

Moreover, the same encoding shows decidability for word equations with length
constraints, provided the proof tree is associated with flat counter systems.
In particular, we show that the class of \emph{regular-oriented} word equations, introduced by \cite{DMN17},
have flat proof graphs.
Thus, the satisfiability problem for quadratic regular-oriented word equations with length constraints 
is decidable. In fact, we obtain an NP algorithm with an oracle access to 
$\PAD$; the best complexity bound for the latter is NEXP and NP-hardness 
\cite{LOW15}.
A standard \defn{monoid technique} for handling regular constraints in word
equations (e.g.~see \cite{diekert-quadratic1}) can then be used to extend our 
decidability result in the presence of regular constraints. This results in a
PSPACE algorithm with an oracle access to $\PAD$.

While our decidability result is for a simple subclass, this class is already non-trivial without length and regular
constraints: satisfiability of regular-oriented word equations is NP-complete \cite{DMN17}.
Notice that in the presence of regular constraints the complexity (without
length constraints, and even without word equations) jumps to PSPACE, by virtue
of the standard PSPACE-completeness of intersection of regular languages
\cite{Kozen77}.
Moreover, we believe that the techniques in this paper --- the connection between acceleration and
word equations, and the use of existential Presburger with divisibility --- can 
pave the way to more sophisticated decision
procedures based on counter system acceleration.

\section{Preliminaries}
\label{sec:prelim}

\noindent
\textbf{General notation:}
Let $\N = \Z_{\geq 0}$ be the set of all natural numbers. For integers $i \leq 
j$, we use $[i,j]$ to denote the set $\{i,i+1,\ldots,j-1,j\}$ of integers.
If $i \in \N$, let $[i]$ denote $[0,i]$. We use $\preceq$ to denote the
component-wise ordering on $\N^k$, i.e., $(x_1,\ldots,x_k) \preceq
(y_1,\ldots,y_k)$ iff $x_i \leq y_i$ for all $i \in [1,k]$. If $\bar x \preceq
\bar y$ and $\bar x \neq \bar y$, we write $\bar x \prec \bar y$.

If $S$ is a set,
we use $S^*$ to denote the set of all finite sequences, or \emph{words}, $\gamma = s_1\ldots s_n$ 
over $S$. The length $|\gamma|$ of $\gamma$ is $n$.
The empty sequence is denoted by $\epsilon$. 
Notice that $S^*$ forms a monoid with the concatenation operator $\cdot$. If $\gamma'$ is a prefix
of $\gamma$, we write $\gamma' \preceq \gamma$. Additionally, if $\gamma' \neq
\gamma$ (i.e. a strict prefix of $\gamma$), we write $\gamma' \prec \gamma$.
Note that the operator $\preceq$ is overloaded here, but the meaning should be
clear from the context.

\smallskip
\noindent
\textbf{Words and automata:} 
We assume basic familiarity with word combinatorics and automata theory.
Fix a (finite) alphabet $\ialphabet$. 
For each finite word $w := w_1\ldots w_n \in
\ialphabet^*$, we write $w[i,j]$, where $1 \leq i \leq j \leq n$, to denote the segment
$w_i\ldots w_j$. 

Two words $x$ and $y$ are \emph{conjugates} if there exist words $u$ and $v$ such that
$x = uv$ and $y = v u$.
Equivalently, $x = \mathrm{cyc}^k(y)$ for some $k$ and for the \emph{cyclic
permutation} operation $\mathrm{cyc} : \ialphabet^* \rightarrow \ialphabet^*$, defined as
$\mathrm{cyc}(\epsilon) = \epsilon$, and $\mathrm{cyc}(a\cdot w) = w\cdot a$ 
for $a\in \ialphabet$ and $w \in \ialphabet^*$.

\lmcschanged{
Given a \defn{nondeterministic finite automaton (NFA)} 
$\Aut := (\ialphabet,\controls,\transrel,q_0,\lmcschanged{\finals})$, where 
$\transrel \subseteq
\controls \times \ialphabet \times \controls$,
a \defn{run} of $\Aut$ on $w = a_1\cdots a_n$ is a function $\AutRun: [0,n]
\to \controls$ with $\AutRun(0) = q_0$ that obeys the transition relation 
$\transrel$: $(\AutRun(i),a_{i+1},\AutRun(i+1)) \in \transrel$, for all $i \in 
[0,n-1]$.}
We may also denote the run $\rho$ by the word $\rho(0)\cdots \rho(n)$ over
the alphabet $\controls$. 
The run $\AutRun$ is said to be \defn{accepting} if \lmcschanged{$\rho(n) \in 
\finals$}, in which
case we say that the word $w$ is \defn{accepted} by $\Aut$. The language
$\Lang(\Aut)$ of $\Aut$ is the set of words in $\ialphabet^*$ accepted by
$\Aut$. In the sequel, for $p,q \in \controls$ we will write $\Aut_{p,q}$ to
denote the NFA $\Aut$ with initial state replaced by $p$ and
final state replaced by $q$.

\smallskip
\noindent
\textbf{Word equations:}
Let $\ialphabet$ be a (finite) alphabet of constants and $\Vars$ a set of 
variables; we assume $\ialphabet \cap \Vars = \emptyset$.
A \emph{word equation} $E$ is an expression of the form $L=R$, where $(L,R) \in
(\ialphabet\cup\Vars)^* \times (\ialphabet\cup\Vars)^*$.
A system of word equations is a nonempty set 
$\set{L_1 = R_1, L_2 = R_2,\ldots, L_k = R_k}$ of word equations.
The length of a system of word equations is the length 
$\sum_{i=1}^k (|L_i|+|R_i|)$.
A system is called \emph{quadratic} if each variable occurs at most twice in 
the entire system.
A \emph{solution} to a system of word equations 
is a homomorphism $\soln : (\ialphabet\cup\Vars)^* \rightarrow \ialphabet^*$
which maps each $a\in \ialphabet$ to itself that equates the l.h.s. and
r.h.s. of each equation, i.e.,
$\soln(L_i) = \soln(R_i)$ for each $i=1,\ldots,k$.

For each variable $x\in\Vars$, we shall use $|x|$ to denote a formal variable 
that stands for the length of variable $x$, i.e.,
for any solution $\sigma$, the formal variable $|x|$ takes the value $|\sigma(x)|$.
Let $L_{\Vars}$ be the set $\set{|x|\mid x\in\Vars}$.
A \emph{length constraint} is a formula in Presburger arithmetic whose free 
variables are in $L_\Vars$.

A \emph{solution} to a system of word equations with a length constraint $\Phi(|{x_1}|, \ldots, |{x_n}|)$
is a homomorphism $\soln : (\ialphabet\cup\Vars)^* \rightarrow \ialphabet^*$
which maps each $a\in \ialphabet$ to itself such that $\soln(L_i) = \soln(R_i)$ for each $i=1,\ldots,k$
and moreover $\Phi(|\soln(x_1)|, \ldots, |\soln(x_n)|)$ holds.
That is, the homomorphism maps each variable to a word in $\ialphabet^*$ such that each word equation
is satisfied, and the lengths of these words satisfy the length constraint.

The \defn{satisfiability problem} for word equations with length constraints 
asks, 
given a system of word equations and a length constraint, whether it has a solution.

We also consider the extension of the problem with regular constraints.
For a system of word equations, a 
variable $x\in\Vars$, and a regular language $\mathcal{L} \subseteq \ialphabet^*$, 
a \emph{regular constraint} $x\in \mathcal{L}$ imposes the additional restriction
that any solution $\soln$ must satisfy $\soln(x)\in \mathcal{L}$.
Given a system of word equations, a length constraint, and a set of regular constraints,
the satisfiability problem asks if there is a solution satisfying the word equation,
the length constraints, as well as the regular constraints.

In this paper we consider only \emph{a system consisting of a single word 
equation}. We note quickly that there are various possible 
satisfiability-preserving 
\lmcschanged{
reductions that reduce a system of word equations to a single 
word equation (e.g. \cite{KMP20})}, but most of them
do not in general preserve desirable properties such as the quadraticity of the 
constraints. 
One simple method that preserves the quadraticity of the
constraints is to simply concatenate the left/right hand sides of the equations
and simply introduce some length constraints. For example, suppose we have
constraints $xy = yz \wedge zz_1 = z_1z_2 \wedge |x| = |z_2|$. The resulting
constraint involving only a single word equation would be
$xyzz_1 = yzz_1z_2 \wedge |x| = |z_2| \wedge |x|+|y| = |y|+|z|$. 
In general, even this reduction does not preserve the properties that we 
consider in this paper (flatness, regularity, orientedness), unless further 
restrictions are imposed. \lmcschanged{See Section \ref{sec:decidability}} for the definitions.
We will remark this further in the relevant parts of the paper.

\smallskip
\noindent
\textbf{Linear arithmetic with divisibility:}
Let $\PA$ be a first-order language \lmcschanged{over the domain $\N$ of natural
numbers}
with equality = and inequality $\leq$ of numbers relations, and with terms 
being linear
polynomials with integer coefficients. 
We write $f(x)$, $g(x)$, etc., for terms in 
integer variables $x = x_1, \ldots , x_n$. 
Atomic formulas in Presburger arithmetic 
have the form $f(x) \leq g(x)$ or
$f(x) = g(x)$.
%
The language $\PAD$ of \emph{Presburger arithmetic with divisibility}
extends the language $\PA$ with a binary relation $\divides$ (for divides).
An atomic formula has the form 
$f(x) \leq g(x)$ or $f(x) = g(x)$ or $f(x) \divides g(x)$,
where $f(x)$ and $g(x)$ are linear polynomials with integer coefficients.
The full first order theory of $\PAD$ is undecidable, but the existential
fragment is decidable \cite{Lipschitz,LOW15}.

Note that the divisibility predicate $x \divides y$ is \emph{not} expressible
in Presburger arithmetic: a simple way to see this is that $\set{(x,y) \in\N^2 \mid x \divides y}$
is not a semi-linear set.

\smallskip
\noindent
\textbf{Counter systems:}
In this paper, we specifically use the term ``counter systems'' to mean
counter systems with Presburger transition relations
(e.g.~see \cite{FAST}). These more general transition relations can be 
simulated by standard Minsky's counter machines, but they are more useful for
coming up with decidable subclasses of counter systems.
A \defn{counter system} $\CA$ is a tuple $(\counters,\controls,\transrel)$, 
where $\counters = \{\ctr_1,\ldots,\ctr_m\}$ is a finite set of counters, 
$\controls$ is a finite set of control states, and $\transrel$ is a finite
set of transitions of the form $(q,\Phi(\bar\ctr,\bar\ctr'),q')$, where
$q,q' \in \controls$ and $\Phi$ is a Presburger formula with free variables
$\ctr_1,\ldots,\ctr_m,\ctr_1',\ldots,\ctr_m'$. A \defn{configuration} of
$\CA$ is a tuple $(q,\vecV) \in \controls \times \N^m$. 

The semantics of counter systems is given as a transition system.
A \defn{transition system} is a tuple $\struct := \transys$,
where $\transysDom$ is a set of \defn{configurations}
and $\to\ \subseteq S \times S$ is a binary relation over $S$.
A \defn{path} in $\struct$ is a sequence $s_0 \to \cdots \to s_n$ of
configurations $s_0,...,s_n \in S$. If $S' \subseteq S$, let $\pre^*(S')$ denote 
the set of $s \in S$ such that $s \to^* s'$ for some $s' \in S'$. 
We might write $\pre_{\to}^*(S')$ to disambiguate the transition system.

A counter system
$\CA$ generates the transition system $\Trans{\CA} = \transys$, where 
$\transysDom$ is the set of all configurations of $\CA$, and 
$(q,\vecV) \to (q',\vecV')$ if there exists a transition
$(q,\Phi(\bar\ctr,\bar\ctr'),q') \in \transrel$ such that
$\Phi(\vecV,\vecV')$ is true. 

In the sequel, we will be needing the notion of flat counter systems
\cite{BIL09,LS06,BFLS05,FAST}. 
Given a counter system $\CA = (\counters,\controls,\transrel)$, the
\defn{control structure} of $\CA$ is an edge-labeled directed
graph $G = (V,E)$ with the set $V = \controls$ of nodes and the set 
$E = \transrel$.
The counter system $\CA$ is \defn{flat} if each node $v \in V$ is contained
in at most one simple cycle. We now define the notion of ``skeletons'' in $\CA$,
which we will use in Section \ref{sec:decidability}. \lmcschanged{Intuitively, 
a skeleton
is simply a simple path with simple cycles along the way in the graph 
$\CA$}.
More precisely, considering $\CA$ a dag of SCCs, we define the 
\defn{signature} $s_\CA$ of $\CA$ as the (directed) graph whose nodes are SCCs 
in $\CA$ and there is an edge from $v$ to $w$ if there is a state in the SCC
$v$ that can go to a state in SCC $w$.
A \defn{skeleton} in $\CA$ is simply a subgraph $G' = (V',E')$ of $\CA$ that
is obtained by taking \lmcschanged{a (simple) path} in $s_\CA$ and expanding 
each node into the corresponding SCC in $\CA$.

\section{Solving Quadratic Word Equations}
\label{sec:nielsen}

We start by recalling a simple textbook recipe (Nielsen transformation, 
a.k.a., Levi's Method) \cite{Diekert,Lentin} for solving quadratic word 
equations, both 
for the cases with and without regular constraints. We then discuss the length
abstractions of solutions to quadratic word equations, and provide a
natural example that is not Presburger-definable.

\subsection{Nielsen transformation}
We will define a rewriting relation $E \To E'$ between quadratic word equations
$E, E'$.
Let $E$ be an equation of the form $\alpha w_1 = \beta w_2$ with 
$w_1,w_2 \in (\ialphabet \cup \Vars)^*$ and $\alpha,\beta \in \ialphabet \cup 
\Vars$. Then, there are several possible $E'$:
\begin{itemize}
    \item \defn{Rules for erasing an empty prefix variable}. These rules can be 
        applied 
        if $\alpha \in \Vars$ (symmetrically, $\beta \in \Vars$). We
        nondeterministically guess that $\alpha$ be the empty 
        word $\epsilon$, i.e., $E'$ is $(w_1 = \beta w_2)[\epsilon/\alpha]$.
       The symmetric case of $\beta \in \Vars$ is similar.
    \item \defn{Rules for removing a nonempty prefix}. These rules are
        applicable if each of $\alpha$ and $\beta$ is either a constant or
        a variable that we nondeterministically guess to be a nonempty word.
        There are several cases:
        \begin{description}
            \item[(P1)] $\alpha \equiv \beta$ (syntactic equality). In this 
                case, $E'$ is $w_1 = w_2$.
            \item[(P2)] $\alpha \in \ialphabet$ and $\beta \in \Vars$. In this case,
                $E'$ is $w_1[\alpha\beta/\beta] = \beta
                (w_2[\alpha\beta/\beta])$. In the sequel, to avoid notational
                clutter we will write $\beta w_2[\alpha\beta/\beta]$ instead of
                $\beta(w_2[\alpha\beta/\beta])$.
            \item[(P3)] $\alpha \in \Vars$ and $\beta \in \ialphabet$. In this case,
                $E'$ is $\alpha (w_1[\beta\alpha/\alpha]) = w_2[\beta\alpha/\alpha]$.
            \item[(P4)] $\alpha,\beta \in \Vars$. In this case, we
                nondeterministically guess if $\alpha \preceq \beta$ or
                $\beta \preceq \alpha$. In the former case,
                the equation $E'$ is $w_1[\alpha\beta/\beta] = \beta (w_2[\alpha\beta/\beta])$.
                In the latter case, the equation $E'$ is
                $E'$ is $\alpha (w_1[\beta\alpha/\alpha]) = w_2[\beta\alpha/\alpha]$. 
        \end{description}
\end{itemize}
Note that the transformation keeps an equation quadratic. 

\begin{proposition}
    $E$ is solvable iff
    $E \To^* (\epsilon = \epsilon)$. Furthermore,
    checking if $E$ is solvable is in PSPACE.
\end{proposition}
    \begin{figure}
    \includegraphics[width=0.5\textwidth,angle = 0]{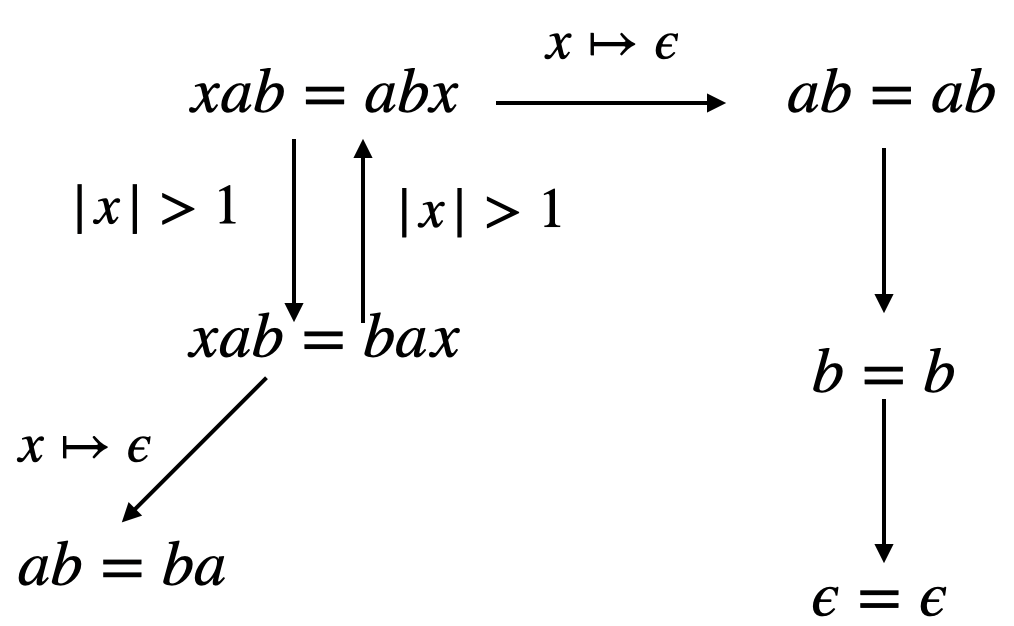}
        \caption{An example of Levi's method applied to a simple example. This
        shows that $xab = abx$ is satisfiable.\label{fig:nielsen_ex}}
    \end{figure}
    \lmcschanged{
The proof is standard (e.g. see \cite{Diekert}). For the sake of completeness,
we illustrate this with an example and provide the crux of the proof.
Figure \ref{fig:nielsen_ex} depicts an
application of Nielsen's transformation to show that $xab = abx$ is satisfiable.
It is not a coincidence that this is a finite graph because the set of reachable
nodes from any quadratic equation $E$ is actually finite. This finiteness 
property essentially can be seen as a corollary of the following two easily
checked properties:
(1) each rule of Nielsen transformation preserves the quadratic nature of an
equation, and (2) each rule of Nielsen transformation \emph{never increases} 
the length of an equation. This directly gives us a PSPACE algorithm: each step
can be implemented in polynomial space, and the size never increases.
The IFF statement in the above proposition can essentially be proven as follows.
The ($\Leftarrow$) direction follows from the fact that each rule of Nielsen
transformation preserves satisfiability. The ($\Rightarrow$) direction follows
from the fact that, when applied to quadratic equations, each step either 
decreases the size of the equation, or the length of a length-minimal solution. 
}


\subsection{Handling regular constraints}
Nielsen transformation easily extends to quadratic word equations with regular
constraints (e.g. see \cite{diekert-quadratic1}). 

Recall that we are given a finite set 
$S = \{x_1 \in \Lang(\Aut^1),\ldots, x_n \in \Lang(\Aut^n)\}$ of regular 
constraints. \lmcschanged{We assume that the automata $\Aut^1,\ldots,\Aut^n$ 
have disjoint state-space.} 
This sequence $x_1,\ldots,x_n$ might contain repetition of
variables. Instead of allowing
multiple regular constraints per variable $x$, we will assign one
\emph{monoid element}
over boolean matrices for each variable $x$ in the following way. Suppose 
$\Aut^i$ has $r_i$ states, and we let $r = \sum_{i=1}^n r_i$. 
Let $\Aut$ be the automaton with $r$ states obtained by taking the 
union of $\Aut^1,\ldots,\Aut^n$. Let $\B$ be the 
usual boolean algebra with two elements $0, 1$.  For any given word
$w \in \ialphabet^*$, we can construct the
\defn{characteristic matrix} $M = (m_{i,j}) \in \B^{r\times r}$ (indexed by the 
states of
$\Aut$ without loss of generality) as follows: for any pair $i,j$ of states,
$m_{i,j} = 1$ iff $w \in \Lang(\Aut_{i,j})$.
Therefore, we can assign a monoid element to each variable $x$ in the following
way.
Take the subsequence $\AutB^1,\ldots,\AutB^m$ of
$\Aut^1,\ldots,\Aut^n$, for which there is a regular constraint $x \in
\Lang(\AutB^i)$. Suppose
\lmcschanged{$(q_0^1,\finals^1),\ldots,(q_0^m,\finals^m)$} are the 
pairs of initial state and set of final states of $\AutB^1,\ldots,\AutB^m$. 
Then,
there exist a homomorphism $\phi: \ialphabet^* \to \B^{r\times r}$,
such that
\lmcschanged{
\[
    w \in \bigcap_{i=1}^{m} \Lang(\AutB^i) \quad \text{iff} \quad 
    \forall i\in\{1,\ldots,m\}: \bigvee_{q_F \in \finals^i} 
        \phi(w)[q_0^i,q_F] = 1.
\]
}
Note that the term homomorphism here is justified since $\B^{r\times r}$ is
a monoid with the boolean matrix multiplication operator. In the following,
we will always restrict ourselves to the submonoid of $\B^{r\times r}$
containing only \emph{realisable characteristic matrices}, i.e., those matrices
$M = \phi(w)$, for some $w \in \ialphabet^*$. Note that checking whether a
boolean matrix is realisable is PSPACE-complete since it is equivalent to
checking emptiness of intersection of automata \cite{Kozen77}.

\OMIT{
Now observe that a monoid element $M \in \B^{r \times r}$ actually carries more
information that the given regular constraints on $x$, which do not necessarily say whether
or not $x \in \Lang(\Aut)_{p,q}$, for \emph{all} pair $p,q$ of states. 
We will
see later that an algorithm in this monoid framework would essentially
nondeterministically guess one of these monoid elements $M$ corresponding to
the regular constraints.
}

\OMIT{
We will think of these
constraints as a 
We assume that a regular
constraint $x \in \mathcal{L}$ is given as an NFA $\Aut_{p,q}$ representing $\mathcal{L}$. 
If $q_0$ and $q_F$ are the initial and final states (respectively) of an
NFA $\Aut$, we can be more explicit and write $\Aut_{q_0,q_F}$ instead of 
$\Aut$. 
}

Our rewriting relation $\To$ now works over a pair $(E,f)$ consisting of a word
equation $E$ and a function mapping each variable $x$ in $E$ to a monoid 
element $M \in \B^{r\times r}$.
Let $E$ be an equation of the form $\alpha w_1 = \beta w_2$ with 
$w_1,w_2 \in (\ialphabet \cup \Vars)^*$ and $\alpha,\beta \in \ialphabet \cup 
\Vars$. We now define $(E,f) \To (E',f')$ by extending the
\lmcschanged{previous} definition of $\To$ 
without regular constraints. 
More
precisely, $(E,f) \To (E',f')$ iff $E \To E'$ and additionally do the following:
\begin{itemize}
    \item \defn{Rules for erasing an empty prefix variable $\alpha$}. When applied,
        ensure that $\phi(\epsilon) = f(\alpha)$. Without loss of generality,
        we may assume that all our automata have no $\epsilon$-transitions,
        in which case $f(\alpha)$ is the identity matrix.
        \OMIT{, i.e.,
        $f(\alpha) \in \Lang(\Aut)_{p,q}$ iff $\epsilon \in
        \Lang(\Aut)_{p,q}$. }
        We define $f'$ as the restriction of $f$ to $\Vars
        \setminus \{\alpha\}$.
    \item  \defn{Rules for removing a nonempty prefix}. For (P1), we 
        set $f'$ to be the restriction of $f$ to $\Vars \setminus \{\alpha\}$.
        For (P2)--(P4), assume that $E'$ is 
        $w_1[\alpha\beta/\beta] = \beta (w_2[\alpha\beta/\beta])$; the other case is symmetric. 
        We nondeterministically guess a monoid element $M_{\beta'}$. If
        $\alpha \in \ialphabet$, we check  that
        $f(\beta) = \phi(\alpha) \cdot M_{\beta'}$.
        If $\alpha \in \Vars$, we make sure that
        $f(\beta) = f(\alpha) \cdot M_{\beta'}$. Note that both of these
        checks can be done in polynomial time.
        We set $f'$ to be the same function $f$, but differs only on $\beta$:
        $f'(\beta) = M_{\beta'}$.
        \OMIT{
        In the case when $\alpha \in \ialphabet$, we
        also ensure that $\phi(\alpha) = f(\alpha)$.
        positive outcome implies removing this constraint from $S'$, while 
        on a negative outcome our algorithm simply fails on this branch.
        For any variable $y$ that is distinct from $\beta$, 
        we add all regular constraints $y \in L$ in $S$ to $S'$. If 
        $\alpha$ still occurs in $E'$, add regular constraints $\alpha \in L$
        in $S$ to $S'$.  If $S'$ is unsatisfiable, fail on this this branch.
    }
\end{itemize}

We now state the main property of the above algorithm. To this end,
a function $f: \Vars \to \B^{r\times r}$ is said to be \defn{consistent} with 
the set $S$ of regular constraints if, whenever \lmcschanged{$(x \in \Lang(\AutB)) \in 
S$ where $\AutB$ has initial state (resp. set of final states) $p$ (resp.
$\finals$), it is the case that $\bigvee_{q \in \finals} M[p,q] = 1$}, 
\lmcschanged{where $M := f(x)$}.
\begin{proposition}
    $(E,S)$ is solvable iff there exists $f: \Vars \to \B^{r\times r}$ 
    consistent with $S$ such that
    $(E,f) \To^* (\epsilon = \epsilon,\emptyset)$, where $\emptyset$ denotes the
    function with the empty domain. 
    Furthermore, checking if $(E,S)$ is solvable is in PSPACE.
\end{proposition}
\begin{figure}
    \includegraphics[width=0.8\textwidth,angle = 0]{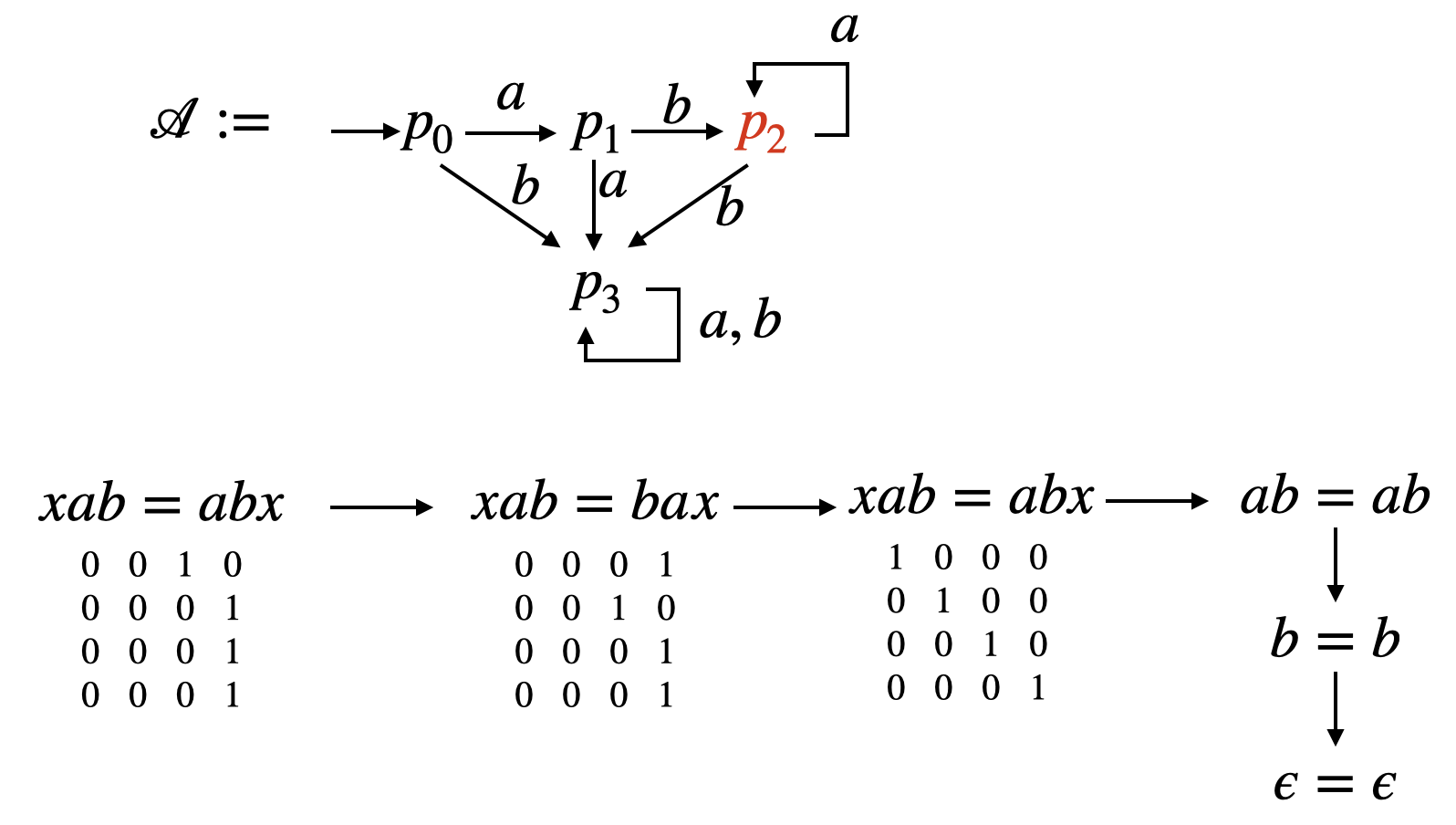}
        \caption{Building on Example from Figure \ref{fig:nielsen_ex}, we
        add the regular constraint $x \in aba^*$ given by automaton 
        $\mathcal{A}$. The automaton is quite large, but we show only the
        path from $(E,f)$ with consistent $f$ to $(\epsilon =
        \epsilon,\emptyset)$. Note that $f$ can be realized by $w = ab$.
        \label{fig:nielsen_regex}}
\end{figure}
See Figure \ref{fig:nielsen_regex} for an example.
Note that the above algorithm is a nondeterministic polynomial-space algorithm
(because of the guessing of $f$ in the above proposition, and also the relation
$\To$), which still gives us a PSPACE algorithm because of the standard 
Savitch's Theorem that PSPACE = NPSPACE \cite{Sipser-book}.
\OMIT{
Note that this is still a PSPACE algorithm because it never creates a new
NFA or adds new states to existing NFA in the regular constraints, but rather 
adds a regular constraint $x \in \mathcal{L}(\Aut_{p,q})$ to a variable $x$, where $\Aut$
is an NFA that is already in the regular constraint.
}

\subsection{Generating all solutions using Nielsen transformation}
One result that we will need in this paper is that Nielsen transformation
is able to \emph{generate all solutions} of quadratic word equations with 
regular constraints. To clarify this, we extend the definition of $\To$ so that
each configuration $E$ or $(E,f)$ in the graph of $\To$ is also annotated by
an assignment $\soln$ of the variables in $E$ to concrete strings. We write 
$E_1[\soln_1] \To E_2[\soln_2]$ if $E_1 \To E_2$ and $\soln_2$ is the
modification from $\soln_1$ according to the operation used to obtain $E_2$ 
from $E_1$. Observe that the domain of $\soln_2$ is a subset of the domain of
$\soln_1$; in fact, some rules (e.g., erasing an empty prefix variable) could
remove a variable in the prefix in $E_1$ from $\soln_1$. The following
example illustrates how $\To$ works with this extra annotated assignment.
Suppose that $\soln_1(x) = ab$ and $\soln_1(y) = abab$ 
and $E_1 := xy = yx$ and $E_2$ is obtained from $E_1$ using rule (P4), i.e.,
substitute $xy$ for $y$. In this case, $\soln_2(x) =
\soln_2(y) = \soln_1(x) = ab$. Observe that $E_2[\soln_2] \To E_3[\soln_3] \To
E_4[\soln_4]$, where $E_3 := E_2$, $\soln_3(x) = ab$, $\soln_3(y) = \epsilon$,
$E_4 := x = x$, and $\soln_4(x) = ab$.
The definition for the case with regular constraints is identical.

\begin{proposition}
    $(E,f)[\soln] \to^* (\epsilon = \epsilon,\emptyset)[\soln']$ where 
    $\soln'$ has the empty domain iff
    $\soln$ is a solution of $(E,f)$.
    \label{prop:all-soln}
\end{proposition}
This proposition immediately follows from the proof of correctness of Nielsen
transformation for quadratic word equations (cf.~\cite{Diekert}).

\subsection{Length abstractions and semilinearity}
\lmcschanged{
Below we tacitly assume an implicit ordering of the set of variables $\Vars = 
\{x_1,\ldots,x_k\}$ --- which ordering is immaterial, so long as one is fixed.}
Given a quadratic word equation $E$ with constants $\ialphabet$ and variables
$\Vars = \{x_1,\ldots,x_k\}$, its \emph{length abstraction} is defined as 
follows
\[
    \LEN(E) = \{ (|\soln(x_1)|,\ldots,|\soln(x_k)|) : \text{$\soln$ is a
        solution to $E$} \},
    \]
namely the set of tuples of numbers corresponding to lengths of solutions to 
$E$.  
\begin{example}
    Consider the quadratic equation $E:=\ xaby = yz$, where $\Vars = \{x,y,z\}$
    and $\ialphabet$ contains at least two letters $a$ and $b$. We will show
    that its length abstraction $\LEN(E)$ can be captured by the Presburger 
    formula $|z| = |x|+2$. Observe that each $(n_x,n_y,n_z) \in \LEN(E)$
    must satisfy $n_z = n_x + 2$ by a length argument on $E$. Conversely,
    we will show that each triple $(n_x,n_y,n_z) \in \N^3$ satisfying
    $n_z = n_x + 2$ must be in $\LEN(E)$. To this end, we will define a solution
    $\soln$ to $E$ such that $(|\soln(x)|,|\soln(y)|,|\soln(z)|) =
    (n_x,n_y,n_z)$. Consider $\soln(x) = a^{n_x}$. 
    Then, for some $q \in \N$ and $r \in [n_x+1]$, we have
    $n_y = q(n_x+2)+r$. Let $w$ be a prefix of $\soln(x)ab$ of 
    length $r$. Therefore, for some $v$, we have $wv = \soln(x)ab$.
    Define $\soln(y) = (\soln(x)ab)^qw$. We then have $\soln(x)ab\soln(y)
    = \soln(y)vw$. Thus, setting $\soln(z) = vw$ gives us a satisfying
    assignment for $E$ which satisfies the desired length constraint.
    \OMIT{
    Therefore, $\soln(z)$ is a 
    prefix of $\soln(x)ab\soln(z)$. That is, for some $v \in \ialphabet^*$,
    it is the case that $\soln(z)v = \soln(x)ab\soln(z)$. Letting 
    $\soln(y) = v$ suffices to make $\soln$ satisfy $E$. Also, $|\soln(y)| =
    |\soln(x)ab\soln(z)| - |\soln(z)| = |\soln(x)| + 2$, which satisfies
    the formula $|x|=|y|+2$. 
}
    \qed
\end{example}
However, it turns out that Presburger Arithmetic is not sufficient for
capturing length abstractions of quadratic word equations.
\begin{theorem}
There is a quadratic word equation whose length abstraction is not
    Presburger-definable.
\end{theorem}
To this end, we show that the length abstraction of $xaby = yabx$, where 
$a,b \in \ialphabet$ and $x,y \in \Vars$, is not Presburger definable.
\begin{lemma}
    The length abstraction $\LEN(xaby = yabx)$ coincides with tuples 
    $(|x|,|y|)$ of numbers satisfying the expression $\varphi(|x|,|y|)$
    defined as:
    \begin{eqnarray*}
        |x| = |y| & \vee & (|x| = 0 \wedge |y| \equiv 0 \pmod{2}) 
                  \vee (|y| = 0 \wedge |x| \equiv 0 \pmod{2}) \\
                  & \vee & (|x|,|y| > 0 \wedge \gcd(|x|+2,|y|+2) > 1)
    \end{eqnarray*}
\end{lemma}
Observe that this would imply non-Presburger-definability: for otherwise, since 
the first three disjuncts are Presburger-definable, the last disjunct would also
be Presburger-definable, which is not the case since the property that two 
numbers are relatively prime is not Presburger-definable.
Let us prove this lemma. Let $S = \LEN(xaby = yabx)$. We first show that given
any numbers $n_x, n_y$ satisfying $\varphi(n_x,n_y)$, there are
solutions $\soln$ to $xaby = yabx$ with $|\soln(x)| = n_x$ and $|\soln(y)|=n_y$.
If they satisfy the first disjunct in $\varphi$ (i.e.,  $n_x = n_y$), 
then set $\soln(x) = \soln(y)$ to an arbitrary word $w \in
\ialphabet^{n_x}$. If they satisfy the second disjunct, then $aby = yab$
and so set $\soln(x) = \epsilon$ and $\soln(y) \in (ab)^*$. The same goes with the
third disjunct, symmetrically. For the fourth disjunct (assuming the first three disjuncts
are false), let $d = \gcd(n_x+2,n_y+2)$. 
Define $\soln(x), \soln(y) \in (a^{d-1}b)^*(a^{d-2})$ so that
$|\soln(\alpha)| = n_\alpha$ for $\alpha \in \Vars$. It follows that
$\soln(x)ab\soln(y) = \soln(y)ab\soln(x)$.

We now prove the converse. So, we are given a solution $\soln$ to $xaby = yabx$
and let $u := \soln(x)$, $v := \soln(y)$. Assume to the contrary that 
$\varphi(|u|,|v|)$ is false and that $u$ and $v$ are the shortest such 
solutions. We have several cases to consider:
\begin{itemize}
    \item $u = v$. Then, $|u| = |v|$, contradicting that $\varphi(|u|,|v|)$
        is false.
    \item $u = \epsilon$. Then, $abv = vab$ and so $v \in (ab)^*$, which implies
        that $|v| \equiv 0 \pmod{2}$. Contradicting that $\varphi(|u|,|v|)$
        is false.
    \item $v = \epsilon$. Same as previous item and that $|u| \equiv 0
        \pmod{2}$.
    \item $|u| > |v| > 0$. Since $\varphi(|u|,|v|)$ is false, we have
        $\gcd(|u|+2,|v|+2) = 1$.
        It cannot be the case that $|u| = |v|+1$ since then,
        comparing prefixes of $uabv = vabu$, the letter at position
        $|u|+2$ would be $b$ on l.h.s. and $a$ on r.h.s., which is a
        contradiction. Therefore $|u| \geq |v|+2$. 
        Let $u' = u[|v|+3,|u|]$, i.e., $u$ but with its prefix of length
        $|v|+2$ removed. By Nielsen transformation, we have
        $u'abv = vabu'$. It cannot be the case that $u' = \epsilon$; for,
        otherwise, $abv = vab$ implies $v\in (ab)^*$ and so $u = vab$, implying
        that 2 divides both $|u|+2$ and $|v|+2$, contradicting 
        that $\gcd(|u|+2,|v|+2) = 1$. Therefore, $|u'| > 0$. Since
        $\gcd(|u'|+2,|v|+2) = \gcd(|u|+2,|v|+2) = 1$, we have a shorter
        solution to $xaby = yabx$, contradicting minimality.
    \item $|v| > |u| > 0$. Same as previous item.
\end{itemize}

%
\OMIT{
we assume that $|\soln(x)| > 1$ (otherwise, it is
trivially true).
Next we can show by induction on $n \in \N$ that if 
$n|\soln(x)| \leq |\soln(z)| < (n+1)|\soln(x)|$, then $|\soln(z)| =
n|\soln(x)|$. The base case is when $n = 0$, which is trivial. So, assume that
$n > 0$. In this case, $\soln(z) = \soln(x)u$ for some $u

let us assume that $|\soln(x)| > 1$ and
$|\soln(z)| > 0$ (otherwise, it
is trivially true). Then, since $\soln(x)[1] = \#$ and $\soln(x)[2] \in (a+b)$,
matching the leftmost letter on both sides of $\soln(xz) = 
\soln(zy)$, we have $\soln(z) \geq 2$, $\soln(z)[1] = \#$ and 
$\soln(z)[2] \in (a+b)$. Since $\soln(y)[1] = \#$ and $\soln(x)[1,|\soln(x)|] 
\in (a+b)^*$, it must be the case that $|\soln(z)| \geq |\soln(x)|$. 
$|\soln(x) = \

We will then
}

\section{Reduction to Counter Systems}
\label{sec:reduce}
In this section, we will provide an algorithm for computing a counter system
from $(E,S)$, where $E$ is a quadratic word equation and $S$ is a set of regular
constraints. We will first describe this algorithm for the case without
regular constraints, after which we show the extension to the case with
regular constraints.

Given the quadratic word equation $E$, we show how to compute a counter 
system $\CA(E) = (\counters,\controls,\transrel)$ such that the following 
theorem holds.

\begin{theorem}
  \label{th:reduce}
    The length abstraction of $E$ coincides with 
 \[
\set{v\in\N^{|\Vars|}\mid (E, v) \in \pre_{\CA(E)}^*(\{\epsilon = \epsilon\} \times \N^{|\Vars|})}
\]
\end{theorem}
Before defining $\CA(E)$, we define some notation.
Define the following formulas:
\begin{itemize}
\item $\ID(\bar x,\bar x') := \bigwedge_{x \in \bar x} x' = x$
\item $\SUB_{y,z}(\bar x,\bar x') := \text{\lmcschanged{$0 <$}} z \leq y \wedge y' = y - z \wedge
        \bigwedge_{x \in \bar x, x \neq y} x' = x$
\item $\DEC_{y}(\bar x,\bar x') := y > 0 \wedge y' = y - 1 \wedge 
        \bigwedge_{x \in \bar x, x \neq y} x' = x$
\end{itemize}
Note that the $\neq$ symbol in the guard of $\bigwedge$ denotes syntactic
equality (i.e. not equality in Presburger Arithmetic). We omit mention of the
free variables $\bar x$ and $\bar x'$ when they are clear from the context.

We now define the counter system. 
Given a quadratic word equation $E$ with constants $\ialphabet$ and variables
$\Vars$, we define a counter system 
$\CA(E) = (\counters,\controls,\transrel)$ as follows. 
The counters $\counters$ will be precisely all variables that appear in $E$, i.e., $\counters := \Vars$.
The control states are precisely all equations $E'$ that can be rewritten
from $E$ using Nielsen transformation, i.e., $\controls := \{ E' : E \To^* E'
\}$. The set $\controls$ is finite (at most exponential in $|E|$) as per our 
discussion in the previous section. 

We now define the transition relation $\transrel$. 
We use $\bar x$ to enumerate $\Vars$ in some order.
Given $E_1 \To E_2$ with $E_1,E_2 \in \controls$, we then
add the transition $(E_1,\Phi(\bar x,\bar x'),E_2)$, where $\Phi$ is defined as
follows:
\begin{description}
    \item[(r1)] If $E_1 \To E_2$ applies a rule for erasing an empty prefix variable
        $y \in \bar x$, then $\Phi := y = 0 \wedge \ID$.
    \item[(r2)] If $E_1 \To E_2$ applies a rule for removing a nonempty prefix:
        \begin{itemize}
            \item If (P1) is applied, then $\Phi = \ID$.
            \item If (P2) is applied, then $\Phi = \DEC_{\beta}$.
            \item If (P3) is applied, then $\Phi = \DEC_{\alpha}$.
            \item If (P4) is applied and $\alpha \preceq \beta$, then $\Phi = 
                \SUB_{\beta,\alpha}$. If $\beta \preceq \alpha$, then
                $\Phi = \SUB_{\alpha,\beta}$.
        \end{itemize}
\end{description}
\begin{lemma}
    The counter system $\CA(E)$ terminates from every configuration
    $(E_0,\vecV_0)$.
\end{lemma}
\lmcschanged{
\begin{proof}
    By checking each transition rule defined for $\CA(E)$,
    it is easy to verify that if $(E_1,\vecV_1) \to (E_2,\vecV_2)$, then $|E_1| 
    \leq |E_2|$ and $\vecV_1 \preceq \vecV_2$. In addition, if $\vecV_1 = 
    \vecV_2$, then $|E_1| < |E_2|$. Indeed, $\vecV_1 = \vecV_2$ could hold
    only in the case of applying (r1) or (r2,P1); the other rules would
    guarantee that $\vecV_2 \prec \vecV_1$ because of the nonzero assumption for
    the decrementing counter. It is clear that (r1) and (r2,P1) would ensure
    that $|E_2| < |E_1|$. This guarantees that it takes at most $|E_0| +
    \|\vecV_0\|$ steps before $\CA(E)$ terminates from $(E_0,\vecV_0)$, where 
    $\|\vecV\|$ denotes the sum of all the components in $\vecV$.
    This completes the proof of the lemma.
\end{proof}
}
The proof of Theorem \ref{th:reduce} immediately follows from Proposition
\ref{prop:all-soln} that Nielsen transformation generates all solutions. 
We illustrate how such a counter system is constructed on the simple example
$xy = yz$ in Figure \ref{fig:nielsen_ex3}.
    \begin{figure}
    \includegraphics[width=0.8\textwidth,angle = 0]{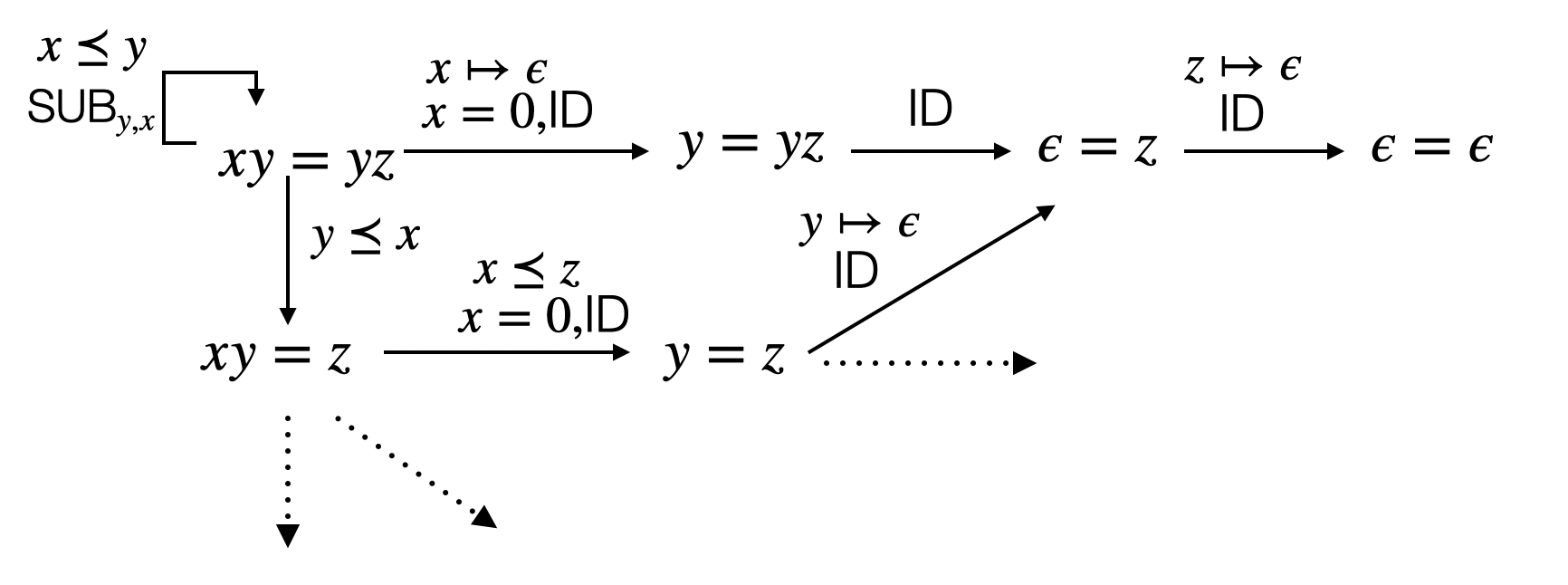}
        \caption{A depiction of parts of $\CA(xy=yz)$. \label{fig:nielsen_ex3}}
    \end{figure}

\subsubsection*{Extension to the case with regular constraints: } 
\OMIT{
As for Nielsen transformation with to quadratic word equations with regular 
constraints, the extension of our reduction is only needed for applications of 
rules (P2)--(P4). 
}
In this extension, we will only need to assert that \emph{initial} counter 
values belong to the length abstractions of the regular constraints, which
are effectively semilinear due to Parikh's Theorem \cite{Parikh}. 
Given a quadratic word equation $E$ with a set $S$ of regular constraints,
we define the counter system $\CA(E,S) = (\counters,\controls,\transrel)$
as follows.
Let $\CA(E) = (\counters_1,\controls_1,\transrel_1)$ be the counter system from the previous
paragraph, obtained by ignoring the regular constraints. 
We define $\counters = \counters_1$. 
Let $\controls$ be the finite set of all configurations reachable from 
some $(E,f)$ where $f$ is a monoid element consistent with $S$, i.e., 
$\controls = \{ (E',f') : (E,f) \To^* (E',f') \}$. Given 
$(E_1,f_1) \To (E_2,f_2)$, we add the transition $((E_1,f_2),\Phi(\bar x,\bar
x'),(E_2,f_2))$ to $\CA(E,S)$ if $(E_1,\Phi(\bar x,\bar x'),E_2)$
was added to $\transrel_1$ by $E_1 \To E_2$. 
The following theorem generalises Theorem \ref{th:reduce} by asserting that the
initial counter values belong to the length abstractions of a monoid element
consistent with the regular constraints.

\begin{theorem}
    The length abstraction of $(E,S)$ coincides with 
    \[
        \bigcup_f
        \left( 
\set{v\in\N^{\Vars}\mid ((E, f), v) \in \pre_{\CA(E,S)}^*(\{(\epsilon =
    \epsilon,\emptyset)\} \times \N^{\Vars})} 
            \cap 
            \LEN_{\bar x}(\phi^{-1}(f))
    \right),
\]
    where $f$ ranges over partial functions mapping $\Vars$ to $\B^{r\times r}$ 
    that are consistent with $S$, and
            $\LEN_{\bar x}(\phi^{-1}(f))$ is a shorthand for
            $\{ v \in \N^{\Vars} \mid v(x) \in \LEN(\phi^{-1}(f(x))),
            \text{\lmcschanged{for all $x \in \Vars$}}\}$.
    \label{th:reduce2}
\end{theorem}
As for the case without regular constraints, the proof of Theorem 
\ref{th:reduce} immediately follows from Proposition
\ref{prop:all-soln} that Nielsen transformation generates all solutions. 
As previously mentioned, $\LEN_{\bar x}(\phi^{-1}(f))$ is semilinear by Parikh's
Theorem \cite{Parikh}. Unlike Theorem \ref{th:reduce}, however, 
Theorem \ref{th:reduce2} is not achieved by a polynomial-time reduction 
for two reasons. Firstly, there could be exponentially many functions $f$ that 
are consistent with $S$, and that checking this consistency is PSPACE-complete.
Secondly, 
the size of the NFA for
$\phi^{-1}(f(x))$ is exponential in the total number of states in automata in
the set $S$ of regular constraints. It turns out that this reduction can be made
to run in polynomial-space. Enumerating the candidate function $f$ can be done
in polynomial space, as previously remarked. 
\OMIT{
It remains to show that
there is a polynomial space algorithm that can enumerate polynomial-size linear 
sets for $\LEN_{\bar x}(\phi^{-1}(f))$.
}

\begin{lemma}
    There is a polynomial-space algorithm which enumerates linear sets 
    corresponding to $\LEN_{\bar x}(\phi^{-1}(f))$.
    \label{lm:parikh_pspace}
\end{lemma}
This lemma essentially follows from the proof of the result of Kopczynski and 
To \cite{KT10} that one can compute a union of $n^{k^{O(1)}}$ linear sets
each with at most $k$ periods with numbers represented in unary
(with time complexity $n^{k^{O(1)}}$, and polynomial space) representing the 
Parikh image of an NFA with $k$ letters in the alphabet. This result can be
easily adapted to show that, given NFAs $\Aut_1,\ldots,\Aut_n$, one can
enumerate in polynomial space a union of exponentially many linear sets 
of polynomial size 
(with at most $k$ periods and numbers represented in binary) representing
the Parikh image of $\Lang(\Aut_1) \cap \cdots \cap \Lang(\Aut_n)$.
Since $\phi^{-1}(f(x))$ can be represented by an intersection of polynomially
many NFA, Lemma \ref{lm:parikh_pspace} immediately follows.
\OMIT{
Furthermore,
since the size of the NFA for
$\phi^{-1}(f(x))$ is exponential in the total number of states in automata in
the set $S$ of regular constraints, one 

. It turns out that this reduction can be made
}
\anthony{The proof of this theorem is not hard, but I will write it up after
I finish the rest.}

\OMIT{
\begin{itemize}
    \item If $E_1 \To E_2$ applies a rule for erasing an empty variable or 
        (P1), then $\Phi := \Phi \wedge \bigwedge_{(x \in L) \in S} x \in 
        \LEN(L)$.
    \item Suppose $E_1 \To E_2$ applies (P2)--(P4). Then, assuming $E_1 :=
        \alpha w_1 = \beta w_2$ and $E_2 := E_1[\alpha\beta/\beta]$, we
        define 
\end{itemize}
}
 
\section{Decidability via Linear Arithmetic with Divisibility}
\label{sec:decidability}

\subsection{Accelerating a 1-variable-reducing cycle}
Consider a counter system $\CA = (\counters,\controls,\transrel)$ with
$\controls = \{q_0,\ldots,q_{n-1}\}$, such that for some $y \in \counters$
the transition relation $\transrel$ consists of precisely the following
transition $(q_i,\Phi_i,q_{i+1 \pmod{n}})$, for each $i \in 
[n-1]$, and each $\Phi_i$ is either $\SUB_{y,z}$ (with $z$ a variable distinct
from $y$) or $\DEC_y$. Such a counter system is said to be a 
\defn{1-variable-reducing cycle}. 

\begin{lemma}
    There exists a polynomial-time algorithm which given a 1-variable-reducing
    cycle $\CA = (\counters,\controls,\transrel)$ and two states $p,q \in
    \controls$ computes \lmcschanged{a} formula $\varphi_{p,q}(\bar x,\bar x')$ in 
    existential Presburger arithmetic with divisibility such that $(p,\vecV) \to_{\CA}^*
    (q,\vecW)$ iff $\varphi_{p,q}(\vecV,\vecW)$ is satisfiable.
    \label{lm:accelerate}
\end{lemma}
This lemma can be seen as a special case of the acceleration lemma for flat 
parametric counter automata \cite{BIL09} (where all variables other than $y$ are
treated as parameters). However, its proof is in fact quite simple.
Without loss of generality, we assume that $q = q_0$ and $p = q_i$, for some
\lmcschanged{$i \in [0,n-1]$}.
Any path $(q_0,\vecV) \to_{\CA}^* (q_i,\vecW)$ can be decomposed into 
the cycle $(q_0,\vecV) \to^* (q_0,\vecV')$ and the simple path $(q_0,\vecW_0)
\to \cdots \to (q_i,\vecW_i)$ of length $i$. Therefore, the reachability
relation $(q_0,\vecX) \to_{\CA}^* (q_i,\vecY)$ can be expressed as
\[
    \exists \vecZ_0,\cdots,\vecZ_{i-1}: \varphi_{q_0,q_0}(\vecX,\vecZ_0) \wedge 
                \Phi_0(\vecZ_0,\vecZ_1) \wedge \cdots \wedge
                \Phi_{i-1}(\vecZ_{i-1},\vecY).
            \]
Thus, it suffices to show that $\varphi_{q_0,q_0}(\vecX,\vecX')$ is expressible
in $\PAD$. 
\rupak{CHECK: rewritten}
Consider a linear expression $M = a_0 + \sum_{x\in X\setminus\set{y}} a_x x$,
where $a_0$ is the number of instructions $i$ in the cycle such that $\Phi_i = \DEC_y$
and $a_x$ is the number of instructions $i$ such that $\Phi_i = \SUB_{y,x}$.
Each time around the cycle, $y$ decreases by $M$.
\rupak{END CHANGE}
Thus, for some $n \in \N$ we have $y' = y - nM$, or 
equivalently
\[
    nM = y - y'
\]
The formula $\varphi_{q_0,q_0}$ can be defined as
follows:
\[
    \varphi_{q_0,q_0} :=\ M \mid (y-y') \wedge 
                        y' \leq y \wedge
                            \bigwedge_{x \in X\setminus\{y\}} x' = x.
\]

\subsection{An extension to flat control structures}

The following generalisation to flat control structures is an easy corollary of 
Lemma \ref{lm:accelerate}. 
\lmcschanged{
\begin{theorem}
    There exists an NP algorithm with a PAD oracle for the following problem:
    given a flat Presburger
    counter system $\CA = (\counters,\controls,\transrel)$, each of whose
    simple cycle is 1-variable-reducing 
    two states $p,q \in \controls$, and an existential Presburger formula
    $\psi(\bar x,\bar y)$, decide if
    $(p,\vecV) \to_{\CA}^* (q,\vecW)$ for some $\vecV,\vecW$ such that
    $\psi(\vecV,\vecW)$ is a valid Presburger formula. That is, the problem
    is solvable in time NEXP.
    \label{th:flat}
\end{theorem}
}
\lmcschanged{
\begin{proof}
    Consider the control graph $G = (V,E)$ of $\CA$. To witness 
    $(p,\vecV) \to_{\CA}^* (q,\vecW)$ for some $\vecV,\vecW$ such that
    $\psi(\vecV,\vecW)$ is true, it is sufficient and necessary that the
    execution of $\CA$ corresponds to a skeleton of $G$ from $p$ to $q$ with a 
    specific entry point (together with a specific incoming transition) and a 
    specific exit point (together with a specific outgoing transition) for 
    each cycle. Namely, this path can be described as:
    \[
        \pi := v_0 \to_E^* v_0' \to_{e_1} v_1 \to_E^* v_1' \to_{e_2} \cdots 
    \to_{e_k} v_k \to_E^* v_k'
    \]
    where 
    \begin{itemize}
        \item $v_0 = p$ and $v_k' = q$ 
        \item $v_i$ and $v_i'$ belong to the same SCC in $G$
        \item $v_i'$ and $v_{i+1}$ belong to different SCCs in $G$.
    \end{itemize}
    In other words, each $v_i$ (resp. $v_i'$) describes the entry (resp. exit)
    node for the $i$th SCC that is used in this path. 
    Our algorithm first guesses this path $\pi$.
    We construct the formula
    $\varphi_{v_i,v_i'}(\bar x_i,\bar x_i')$ by means of Lemma \ref{lm:accelerate}.
    We then only need to use a PAD solver to check satisfiability for the 
    formula
    \[
        \Phi := \left(\bigwedge_{i=0}^k \varphi_{v_i,v_i'}(\bar x_i,\bar
        x_i')\right)
        \wedge \left(\bigwedge_{i=1}^k \theta_i(\bar x_{i-1}',\bar x_i)\right)
        \wedge \psi(\bar x_1,\bar x_k'),
    \]
    where $\theta_i$ applies the transition $e_i$ to $\bar x_{i-1}'$ to obtain
    $\bar x_i$. That is, we have $e_i = (v_i',\theta_i(\bar x,\bar x'),v_i)$.
    Altogether, we obtain an NP algorithm with an access to PAD oracle for the 
    problem. The NEXP upper bound \cite{LOW15} for PAD gives us also an 
    NEXP upper bound for our problem.
\end{proof}
}

\OMIT{
Theorem \ref{th:flat} gives rise to a straightforward \emph{general loop 
acceleration scheme $\ACCELERATE$} 
for underapproximating the solutions to a quadratic word equation $(E,S)$ with
regular constraints. First, we build the counter system $\CA(E,S) =
(\counters,\controls,\transrel)$. Next, a \defn{flattening} of $\CA(E,S)$ is
a flat counter system $\CA' = (\counters,\controls',\transrel')$, each of
whose simple cycle is 1-variable-reducing with unary Presburger guards,
such that there exists a mapping $g: \controls' \to \controls$ (called
\defn{folding}) such that: (1) $(p,\Phi,q) \in \transrel'$ implies
$(g(p),\Phi,g(q)) \in \transrel$, (2) there exists $p,q \in \controls'$
such that $g(p) = (E,S)$ and $g(q) = (\epsilon=\epsilon,\emptyset)$. Therefore,
$\ACCELERATE$ enumerates all flattenings $\CA_1,\CA_2,\ldots$ of $\CA(E,S)$ in 
some order, and 

Next, given
any positive integer $k$, we define the $k$-unfolding of $\CA(E,S)$ as the
counter system $\CA^k(E,S) = (\counters,\controls_k,\transrel_k)$, where:
\begin{itemize}
    \item $\controls_k$ is the set of all paths in the control structures of
        $\CA(E,S)$ of length at most $k$.
    \item 
\end{itemize}
}

\lmcschanged{
\begin{remark}
    \label{rm:flat}
    Note that the reduction from system of word equations to a single word
    equation that we remark in Section \ref{sec:prelim} may preserve flatness,
    but that is not necessarily the case. The positive example where flatness is
    preserved is given in Section \ref{sec:prelim}, but an example where flatness is
    not preserved is $xy = yz \wedge z = x$.
\end{remark}
}

\subsection{Application to word equations with length constraints}

Theorem \ref{th:flat} gives rise to a simple and sound (but not complete)
technique for solving quadratic word equations with length constraints:
given a quadratic word equation $(E,S)$ with regular constraints, if
the counter system $\CA(E,S)$ is flat, each of whose simple cycle is
1-variable-reducing with unary Presburger guards, then apply the decision
procedure from Theorem \ref{th:flat}. In this section, we show completeness of
this method for the class of regular-oriented word equations recently defined 
in \cite{DMN17}, which can be extended with regular constraints given as 1-weak
NFA \cite{BRS12}.
\OMIT{
A word $w\in (\ialphabet\cup\Vars)^*$ is \emph{regular-(ordered)} if each 
variable $x\in\Vars$ appears at most once
in $w$.
The word $w\in (\ialphabet\cup\Vars)^*$ is \emph{non-crossing} if between any two occurrences of the same variable $x$,
no other variable different from $x$ occurs.
For example $ax_1ax_2$ is regular but $ax_1bx_1$ is not, and $ax_1bx_1x_2ax_2$ is non-crossing but 
$ax_1bx_2x_1ax_2$ is not.
}
A word equation is \defn{regular} if each variable $x
\in\Vars$ occurs at most once on each side of the equation. 
Observe that $xy = yx$ is regular, but $xxyy = zz$ is not. It is easy to
see that a regular word equation is quadratic.
A word equation $L = R$ is said to be 
\defn{oriented} if there is a total ordering $<$ on $\Vars$ such that 
the occurrences of variables on each side of the equation preserve $<$, i.e.,
if $w = L$ or $w = R$ and $w = w_1\alpha w_2\beta w_3$ for some $w_1,w_2,w_3 \in
(\ialphabet\cup\Vars)^*$ and $\alpha,\beta \in \Vars$, then $\alpha < \beta$.
Observe that $xy = yz$ (i.e. that $x$ and $z$ are conjugates) is oriented, but
$xy = yx$ is not oriented. It was shown in \cite{DMN17} that the satisfiability
for regular-oriented word equations is NP-hard.
We show satisfiability for this class with length constraints is decidable.

\begin{theorem}
    The satisfiability problem of regular-oriented word equations with
    length constraints is decidable in nondeterministic exponential time.
    In fact, it is solvable in NP with PAD oracles.
    \label{th:sat}
\end{theorem}
This decidability (in fact, an NP upper bound) for the
\emph{strictly regular-ordered} subcase, in which each variable occurs precisely
once on each side, was proven in \cite{Ganesh-boundary}. 
For this subcase, it was shown that Presburger Arithmetic is sufficient,
but the decidability for the general class of regular-oriented word equations 
with length constraints remained open.

We start with a simple lemma that $\To$ preserves regular-orientedness. 
\OMIT{
Its proof can be found in the \shortlong{full version}{appendix}.
}
\begin{lemma}
    If $E \To E'$ and $E$ is regular-oriented, then $E'$ is also
    regular-oriented.
    \label{lm:preserve}
\end{lemma}
\begin{proof}
    It is easy to see each rewriting rule preserves regularity. Now, because
    $E'$ is regular, to show that $E' := L = R$ is also oriented it is
    sufficient and necessary to 
    show that there are no two variables $x, y$ such that $x$ occurs before
    $y$ in $L$, but $y$ occurs before $x$ in $R$. All rewriting rules except 
    for (P2)--(P4) are easily seen to preserve orientedness. Let us write
    $E := \alpha w_1 = \beta w_2$ with $\alpha \neq \beta$, and assume $E' = 
    E[\alpha\beta/\beta]$; the
    case of $E' = E[\beta\alpha/\alpha]$ is symmetric. So, $\beta$ is some 
    variable $y$. If $\beta$ does not occur in $w_1$, then $L = 
    w_1$ and $R = \beta w_2$ and that $E$ is oriented
    implies that $E'$ is oriented. So assume that $\beta$ appears in $w_1$,
    say, $w_1 = u\beta v$. Then, $R= \beta w_2$ and $L = u\alpha\beta v$.
    Thus, if $\alpha \in \ialphabet$, $E'$ is oriented because we can use
    the same variable ordering that witnesses that $E$ is oriented. So, assume
    $\alpha \in \Vars$. It suffices to show that $\alpha$ occurs 
    \emph{at most once} in $E$.
    For, if $\alpha$ also occurs on the other side of the equation $E$
    (i.e. in $w_2$), $\alpha$ precedes $\beta$ on l.h.s. of $E$, while
    $\beta$ precedes $\alpha$ on r.h.s. of $E$, which would show that $E$
    is not oriented. 
\end{proof}
Next, we show a bound on the lengths of cycles and paths of the counter system
associated with a regular-oriented word equation.

\begin{lemma}
    \label{lm:rowqIsFlat}
Given a regular-oriented word equation $E$, the counter system
    $\CA(E)$ is flat, \lmcschanged{where each simple cycle is
    1-variable-reducing}. Moreover, the length of each simple 
    cycle (resp.~path) in the control structure of $\CA(E)$ is 
    $O(|E|)$ (resp.~$O(|E|^2)$).
\end{lemma}
As an example, the reader can confirm the flatness of $\CA(xy = yz)$ by 
studying the counter system in Figure \ref{fig:nielsen_ex3}.
\OMIT{
\begin{lemma}
    Given a simple cycle $E_0 \To E_1 \To \cdots \To E_n$ with $n > 0$ (i.e. 
    $E_0 = E_n$ and $E_i \neq E_j$ for all $0 \leq i < j < n$). Suppose
    that $E_i := L_i = R_i$ with $L_i = \alpha_i w_i$ and $R_i = \beta_i w_i'$. 
\end{lemma}
}
\begin{proof}
We now show Lemma \ref{lm:rowqIsFlat}.
Let $E := L = R$. 
\lmcschanged{We first analyze the shape of a simple cycle, from which we
will infer most of the properties claimed in Lemma \ref{lm:rowqIsFlat}.}
Given a simple cycle
$E_0 \To E_1 \To \cdots \To E_n$ with $n > 0$ (i.e. $E_0 = E_n$ and $E_i \neq 
E_j$ for all
$0 \leq i < j < n$), it has to be the case that each rewriting in this cycle
applies one of the (P2)--(P4) rules since the other rules reduce the size of the
equation. We have $|E_0| = |E_1| = \cdots = |E_n|$.
Let $E_i := L_i = R_i$ with $L_i = \alpha_i w_i$ and $R_i = \beta_i w_i'$. 
Let us assume that $E_1$ be $w_0[\alpha_0\beta_0/\beta_0] = 
\beta_0 w_0'[\alpha_0\beta_0/\beta_0]$; the case with
$E_1$ be $\alpha_0 w_0[\beta_0\alpha_0/\alpha_0] = 
w_0'[\beta_0\alpha_0/\alpha_0]$ will be easily seen to be symmetric. 
This assumption implies that $\beta_0$ is a variable $y$, and that 
$L_0 = uyv$ for some words $u,v \in (\ialphabet \cup \Vars)^*$ (for, otherwise,
$|E_1| < |E_0|$ because of regularity of $E$).
Furthermore, 
it follows that, for each $i \in [n-1]$, $E_{i+1}$ is
$w_i[\alpha_iy/\lmcschanged{y}] = y w_i'$ and $\beta_i = y$, i.e., 
the counter system $\CA(E)$ applies
either $\SUB_{y,x}$ (in the case when $x = \alpha_i$) or $\DEC_y$ (in the
case when $\alpha_i \in \ialphabet$). For, otherwise, taking a minimal $i \in
[1,n-1]$ with $E_{i+1}$ being $\alpha_i w_i[y\alpha_i/\alpha_i]
= w_i'[y\alpha_i/\alpha_i]$ for some variable $x = \alpha_i$
shows that $E_i$ is of the form $x ... y ... = y ... x ...$ (since $|E_{i+1}| =
|E_i|$) contradicting that $E_i$ is oriented. 
\lmcschanged{This shows
that this simple cycle is 1-variable-reducing.} Consequently, we have
\begin{itemize}
    \item $R_i = R_j$ for all $i,j$, and 
    \item $L_i = \mathrm{cyc}^i(u)yv$ for all $i \in [n]$
\end{itemize}
\lmcschanged{
The shape of the equations in this simple cycle guarantees also that
any edge out of this simple cycle \emph{reduces the length} of the equation, 
which guarantees that each of the nodes $E_1,\ldots,E_n$ only has this as the
unique cycle. Finally, the length of the cycle is at most $|L_0| - 1 \leq 
|L| - 1$.
%
Since this simple cycle is taken arbitrarily, it follows that the structure of 
$\CA(E)$ is flat, each of whose simple cycle is 1-variable-reducing and is 
of length at most $N = \max\{|L|,|R|\} - 1$. 
}

Consider the signature (i.e. dag of SCCs) of the control structure $\CA(E)$;
see Preliminaries for the definition of ``signature''. In this dag,
each edge from one SCC to the next is size-reducing. 
Therefore, the maximal length of a path in 
this dag is $|E|$. Therefore, since the maximal path of each SCC is $N$ (from
the above analysis), the maximal length of a simple path in the control
structure is at most $N^2$.
\end{proof}
\lmcschanged{
This Lemma implies Theorem \ref{th:sat} for the following reason. We
nondeterministically guess one skeleton of $\CA(E)$ from initial node
$E := L = R$ to $\epsilon = \epsilon$. By Lemma \ref{lm:rowqIsFlat}, this is of
polynomial-size and so we can apply Theorem \ref{th:flat} to obtain a
nondeterministic polynomial-time algorithm with a PAD oracle.
}
\OMIT{
Since this is of
polynomial-size, we can simply apply Theorem \ref{th:flat} to obtain Theorem
\ref{th:sat}. More precisely,
by Theorem~\ref{th:flat}, we obtain the formula $\lambda_{E,(\epsilon =
\epsilon)}(\bar x,\bar x')$. By Theorem~\ref{th:reduce}, the length abstraction
of $\CA(E)$ is $\psi(\bar x) := \exists \bar x' \lambda_{E,(\epsilon = \epsilon)}(\bar x,\bar 
x')$. Therefore, solving the satisfiability of the word equation $E$ with
the length constraint $\theta(\bar x)$, it suffices to ask the query
$\psi(\bar x) \wedge \theta(\bar x)$ to a PAD solver. All in all, the complexity
is NP with a PAD oracle.
}

\lmcschanged{
\begin{remark}
    Note that the same example in Remark \ref{rm:flat} shows that the reduction
    from a system of word equations --- where each equation is regular-oriented
    --- to a single word equation
    does not preserve regular-orientedness. It might be possible to extend the
    notion of regular-orientedness appropriately to a system of word equations
    (similar to the extension of the definition of regular word equations to a
    regular system of word equations given in \cite{DM20}). We leave this for
    future work.
\end{remark}
}

\smallskip
\noindent
\textbf{Handling regular constraints: } 
We now extend Theorem \ref{th:sat} with regular constraints.
\begin{theorem}
    The satisfiability problem of regular-oriented word equations with
    regular constraints and length constraints is solvable in 
    nondeterministic exponential time. In fact, it is solvable in 
    PSPACE with PAD oracles. \label{th:sat_reg}
\end{theorem}
Let us prove this theorem. The first key lemma to prove this theorem is the 
following:
\begin{lemma}
The counter system $\CA(E,S)$ is flat, where each cycle is 1-variable-reducing. 
    Moreover, the length of each simple 
    cycle in the control structure of $\CA(E,S)$ is 
    exponential in $|E|$. The maximal length of a skeleton in the control
    structure of $\CA(E,S)$ is the same as the maximal length of a skeleton
    in the control structure of $\CA(E)$, which is polynomial in $|E|$.
    \label{lm:rowqIsFlat2}
\end{lemma}
Note that this does not immediately imply the complexity upper bound of PSPACE 
with PAD oracles (which we will show later below), but it does imply
decidability in the same way as Lemma \ref{lm:rowqIsFlat} implies 
Theorem \ref{th:sat}. More precisely,
we nondeterministically guess a skeleton in $\CA(E,S)$, say, starting from
$(E,f)$, where $f$ is a monoid element consistent with $S$.
By Theorem \ref{th:flat}, we obtain the PAD formula $\psi(\bar x,\bar x') := \lambda_{(E,f),(\epsilon =
\epsilon,\emptyset)}(\bar x,\bar x')$. If the length constraint is $\theta(\bar
x)$, following Theorem \ref{th:reduce2} and Lemma \ref{lm:parikh_pspace}, our 
polynomial-space algorithm enumerates linear sets $\eta(\bar x)$ representing
the semilinear set $\LEN_{\bar x}(\phi^{-1}(f))$. 
\anthony{I think I haven't defined linear sets yet}
We will then ask the query
\[
    \theta(\bar x) \wedge \eta(\bar x) \wedge \exists \bar x'\psi(\bar x,\bar x')
\]
to the PAD solver. 

\begin{proof}[Proof of Lemma~\ref{lm:rowqIsFlat2}]
By Lemma \ref{lm:rowqIsFlat}, we know that $\CA(E)$ is flat. 
So, suppose that $\CA(E,S)$ is not flat. Therefore, we may take two different 
cycles $\pi, \pi'$ both visiting some node $(E,f)$. 
We write $\pi: (E_0,f_0) \to \cdots \to (E_n,f_n)$, and
$\pi': (E_0',f_0') \to \cdots \to (E_m',f_m')$, where $(E_0,f_0) = (E_0',f_0') =
(E_n,f_n) = (E_m',f_m') = (E,f)$. Projecting to the first argument, it must be
the case that $C = E_0 \to \cdots \to E_n$ and $C' = E_0' \to \cdots \to E_m'$ 
are the
same cycle in $\CA(E)$, \emph{repeated} several times. 
\lmcschanged{In particular, this means
that, for some counter $y$, each counter instruction in these cycles are either 
of the form $\DEC_y$ or $\SUB_{y,x}$.}
Without loss of 
generality,
let us assume that $n \leq m$ (otherwise, we swap $\pi$ and $\pi'$). 
We may assume that $\pi$ and $\pi'$ are length-minimal. We will show
that $n = m$, and that $f_i = f_i'$ for all $i$. 
Since $C$ and $C'$ are the same cycle in $\CA(E)$, 
there is a unique 
sequence of monoid elements $M_0,\ldots,M_{m-1}$ such that:
\begin{enumerate}
    \item[(1)] \lmcschanged{$f_i(y) = M_i  \cdot f_{i+1}(y)$} for all $i \in \{0,\ldots,n-1\}$, and
    \item[(2)] \lmcschanged{$f_i'(y) = M_i \cdot f_{i+1}'(y)$} for all $i \in \{0,\ldots,m-1\}$.
\end{enumerate}
This implies that
\lmcschanged{
\begin{eqnarray*}
    f_0(y) & = & \left(\prod_{i=0}^{n-1} M_i\right) \cdot f_0(y) \\
    f_0(y) & = & \left(\prod_{i=0}^{m-1} M_i\right) \cdot f_0(y)
\end{eqnarray*}    
}
implying that both $\prod_{i=1}^{n-1} M_i$ and \lmcschanged{$\prod_{i=1}^{m-1}
M_i$}
are the (unique) identity matrix $I$. This implies that if $m > n$, we would
have that 
\lmcschanged{
\[
    f_n'(y) = \left(\prod_{i=n}^{m-1} M_i\right) \cdot f_0(y)
\]
}
which contradicts minimality of the cycle $\pi'$. Therefore, we have $n=m$.
Since $f_n = f_n' = f$, by the equations in (1) and (2) we have that 
$f_i = f_i'$ for all $i \in \{0,\ldots,n\}$. This proof also contradicts the
fact that $\pi$ and $\pi'$ are different cycles. Therefore, $\CA(E,S)$ is flat,
and we also conclude from the above proof that each cycle is
1-variable-reducing.

Let us now analyze the length of the cycles and paths in $\CA(E,S)$. Note that
in the above argument $n$ is at most the size of the monoid $\B^{n\times n}$,
which is exponential in $n$. Observe that taking a projection of each skeleton 
in $\CA(E,S)$ to the first component (i.e. omitting the monoid elements)
gives us a unique skeleton in $\CA(E)$. The converse is of course also true:
given a skeleton $\pi$ in $\CA(E)$, there is at least one skeleton $\pi'$ in 
$\CA(E,S)$ whose projection to the first component corresponds to $\pi$. This
shows that the maximal length of skeletons in $\CA(E,S)$ coincides with the
maximal length of skeletons in $\CA(E)$.
\end{proof}

To conclude the complexity bound, it suffices to provide a polynomial-space 
algorithm which accelerates the exponentially-sized cycles in the control 
structure of $\CA(E,S)$.
\begin{lemma}
    There is a polynomial-space algorithm which given two control states $p := 
    (E,f), q := (E',f')$ in a cycle in $\CA := \CA(E,S)$ computes a formula
    $\varphi_{p,q}(\bar x,\bar x')$ in PAD such that $(p,\vecV) \to^*_{\CA}
    (q,\vecW)$ iff $\varphi_{p,q}(\vecV,\vecW)$ is satisfiable.
    \label{lm:accelerate2}
\end{lemma}
\lmcschanged{
\begin{proof}
Firstly, observe that checking that $p$ and $q$ are in the same cycle $\pi$ can 
be checked in polynomial space. 
The proof of this lemma is exactly the same as the proof of Lemma
\ref{lm:accelerate}, but with one important difference: we store the
coefficients for each variable in \emph{binary}, and we find these coefficients
in polynomial space.

    \OMIT{
    ingredient that we will outline
below. If we naively compute the linear expression $M$ as in the proof of Lemma
\ref{lm:accelerate}, we will result in an exponentially large formula. The
important observation is that the cycle $\pi$ is a repetition of a simple cycle 
$C$ in $\CA(E)$, repeated a certain number $t$ of times. This number $t$ is
exponential if written in unary, but \emph{is polynomial} if written in binary.
Therefore, we will instead compute the linear expression $tM$, where $M$ is
computed as in the proof of Lemma \ref{lm:accelerate} for a cycle in $\CA(E)$
and $t$ is determined in polynomial space by going around the cycle $\pi$ once.
    }

    We elaborate the details below.
    Assume that the cycle has length $n$ and is of the form $q_0 \to q_1 \to 
    \cdots \to q_{n-1} \to q_0$, where $q_0 = p$ and $q_j = q$ for some $j \in
    [n-1]$.
    Assume further that the $i$th transition is 
    $(q_i,\Phi_i,q_{i+1 \pmod{n}})$, for each $i \in 
[n-1]$, and each $\Phi_i$ is either $\SUB_{y,z}$ (with $z$ a variable distinct
from $y$) or $\DEC_y$. 
    Note that we are \emph{not} enumerating these nodes explicitly, but instead
    we will do this on the fly by means of a polynomial-space machine.
    The desired formula is of the form
    \[
        \exists \vecZ_0,\vecZ_1: \varphi_{p,p}(\vecX,\vecZ_0) \wedge 
                \Phi(\vecZ_0,\vecZ_1).
            \]
We first define $\varphi_{p,p}(\vecX,\vecX')$. 
Compute a linear expression $M = a_0 + \sum_{x\in X\setminus\set{y}} a_x x$,
where $a_0$ is the number of instructions $i$ in the cycle from $p$ to $p$ such 
    that $\Phi_i = \DEC_y$
and $a_x$ is the number of instructions $i$ such that $\Phi_i = \SUB_{y,x}$.
This can be computed easily by a polynomial machine that keeps $|X|$ counters.
Each time around the cycle, $y$ decreases by $M$.
Thus, for some $n \in \N$ we have $y' = y - nM$, or 
equivalently
\[
    nM = y - y'
\]
The formula $\varphi_{p,p}$ can be defined as
follows:
\[
    \varphi_{p,p} :=\ M \mid (y-y') \wedge 
                        y' \leq y \wedge
                            \bigwedge_{x \in X\setminus\{y\}} x' = x.
\]

The formula $\Phi(\vecX,\vecX')$ can also be defined in a similar way, but we
look at a simple path from $p$ to $q$. More precisely, 
compute a linear expression $M' = a_0' + \sum_{x\in X\setminus\set{y}} a_x' x$,
where $a_0'$ is the number of instructions $i$ in the simple path from $p$ to 
    $q$ such that $\Phi_i = \DEC_y$
and $a_x'$ is the number of instructions $i$ such that $\Phi_i = \SUB_{y,x}$.
The formula $\Phi$ can be defined as follows:
\[
    \Phi :=\ y' = y - M' \wedge
                            \bigwedge_{x \in X\setminus\{y\}} x' = x.
    \tag*{\qedhere}
\]
\end{proof}
}

We conclude this section with the following complementary lower bound
relating to expressivity of length abstractions of regular-oriented word
equations with regular constraints. 
\begin{proposition}
    There exists a regular-oriented word equation with regular constraints
    whose length abstraction is not Presburger.
    \label{prop:nonPres_reg}
\end{proposition}
\begin{proof}
    \lmcschanged{
Take the regular-oriented word equation $xy = yz$ over the alphabet 
$\set{a,b,\#}$, together with regular constraints $x,y \in \#(a+b)^*$.
    }
We claim that its length abstraction is
precisely the set of triples $(n_x,n_y,n_z) \in \N^3$ satisfying the formula
\[
    \varphi(l_x,l_y,l_z) := l_x = l_y \wedge l_x > 0 \wedge l_x \mid l_z.
\]
Since divisibility is not Presburger-definable,
the theorem immediately follows. 
To show that for each triple $\bar n = (n_x,n_y,n_z)$ satisfying $\varphi$
there exists a solution $\soln$ to $E$ and the constraint
$x,y \in \#(a+b)^*$, simply consider $\soln$ with 
$\soln(x) = \soln(y) = \#a^{l_x - 1}$, and
$\soln(z) = \soln(z) = \soln(x)^{n_z/n_x}$. Conversely, consider a solution
$\soln$ satisfying $xz = zy$ and $x,y \in \#(a+b)^*$. We must have $x = y$
since two conjugates $x,y \in \#(a+b)^*$ must apply a full cyclical permutation,
i.e., the same words. 
We then have
$|\soln(x)| = |\soln(y)| > 0$. To show that
$|\soln(x)| \mid |\soln(z)|$, let $|\soln(z)|= q|\soln(x)| + r$ for some
$q \in \N$ and $r \in [|\soln(x)|-1]$. It suffices to show that $r = 0$.
To this end, matching both sides of $E$, we obtain $z = x^qw$, where $w$ is a
prefix of $\soln(x)$ of length $r$. If $r > 0$, then matching both sides of the
equation from the \emph{right} reveals that the last $|\soln(y)|-1$ letters
on l.h.s. of $\soln(E)$ contains $\#$, which is not the case on r.h.s. of
$\soln(E)$, contradicting that $\soln$ is a solution to $E$. Therefore, $r = 0$,
proving the claim.
\end{proof}
This shows that Presburger with divisibility is crucial for this fragment. We
leave as an open problem whether the same lower bound holds even when regular 
constraints
are not imposed.
\section{Conclusion}
In this paper we revisit the problem of word equations with length constraints.
We show that there is a tight connection between word equations and Presburger
with divisibility constraints (PAD). More precisely, the usual Nielsen
transformation for quadratic word equations can be adapted to incorporate length
constraints, whereby a variant of counter machines (instead of a finite graph)
is generated to represent the proof graph. Through this connection, we have
obtained decidability in the case when this counter machine contains only
1-variable-reducing cycles; this is achieved via acceleration by means of
PAD constraints. We have applied this to a class of
word equations considered in the literature called regular-oriented word
equations, and obtain its decidability in the presence of length constraints.
We have also shown that our result can be easily adapted to incorporate
regular constraints using the standard monoid techniques. For these results,
we have achieved close-to-optimal complexity: NP and PSPACE with oracles to PAD
for regular-oriented word equations with length constraints, respectively,
with and without regular constraints. Note that without length constraints
these problems are already NP-complete \lmcschanged{\cite{DMN17} and PSPACE-complete
\cite{diekert-quadratic1,diekert-quadratic2}, respectively.}

There are many future research directions. The big open question is of course
whether we can use the technique in this paper to prove decidability of
word equations with length constraints. Similarly, can we show that the 
length abstractions of quadratic word equations are necessarily definable in
PAD? Using our technique in this paper, it is easy to show that the set of
solutions of every PAD formula can be captured by some (not necessarily
quadratic) word equations. Perhaps, an easier question is whether one can
discover other interesting classes of word equations with length constraints
that can be proven decidable using our technique (or something similar). 
\lmcschanged{One candidate fragment is the class of regular (but not necessarily
oriented) word equations, e.g., $xyz = zyx$. Recently, Day and Manea \cite{DM20}
proved an amazing result that regular word equations without regular constraints
are NP-complete, which was achieved by also analyzing the graph structure
generated by Nielsen's transformation. It is, however, not clear the result 
could be extended when length constraints are incorporated, especially because
the resulting counter system is no longer flat (even for $xy = yx$). To this 
end, new acceleration techniques for counter systems must therefore be 
developed, which can handle non-flat counter systems and do provide at least
existential Presburger Arithmetic formulas with divisibility. 
}

\subsubsection*{Acknowledgment.} 
We thank Jatin Arora, Dmitry Chistikov, Volker Diekert, 
Matthew Hague, Artur Je\.{z}, Philipp R\"{u}mmer, James Worrell, as well as
anonymous LMCS reviewers for the helpful feedback.
This research was partially funded by the ERC Starting Grant AV-SMP (grant
agreement no. 759969), the ERC Synergy Grant IMPACT (grant agreement no.
610150), and the Max-Planck Fellowship.

\bibliographystyle{alpha}
\bibliography{references}

%
%

\end{document}